\documentclass[sigconf]{acmart}

\usepackage{amssymb}
\usepackage{xcolor}
\usepackage{soul}
\usepackage[utf8]{inputenc}
\usepackage{graphicx}
\usepackage{booktabs} 
\usepackage{algorithmic}
\usepackage{algorithm}
\usepackage{epsfig}
\usepackage{amsfonts}
\usepackage{subfigure}
\usepackage{multirow}
\usepackage{mathrsfs} 
\usepackage{url}
\PassOptionsToPackage{hyphens}{url}
\usepackage{footnote}
\usepackage{balance}
\usepackage{fixmath}
\usepackage{bm}
\usepackage{natbib}
\usepackage{mathtools}

\setcopyright{none}

\begin{document}
\title{Novel Adaptive Algorithms for Estimating Betweenness,\\ Coverage and $k$-path Centralities}

\author{Mostafa Haghir Chehreghani}
\affiliation{%
  \institution{LTCI, T\'el\'ecom ParisTech}
  \streetaddress{Universit\'e Paris-Saclay}
  \city{Paris} 
  \postcode{75013}
}
\email{mostafa.chehreghani@gmail.com}

\author{Albert Bifet}
\affiliation{%
  \institution{LTCI, T\'el\'ecom ParisTech}
  \streetaddress{Universit\'e Paris-Saclay}
  \city{Paris} 
  \postcode{75013}
}
\email{albert.bifet@telecom-paristech.fr}

\author{Talel Abdessalem}
\affiliation{%
  \institution{LTCI, T\'el\'ecom ParisTech}
  \streetaddress{Universit\'e Paris-Saclay}
  \city{Paris} 
  \postcode{75013}
}
\email{talel.abdessalem@telecom-paristech.fr}

\begin{abstract}
An important index
widely used to analyze social and information networks is {\em betweenness centrality}.
In this paper, first given a directed network $G$ and a vertex $r\in V(G)$,
we present a novel adaptive algorithm
for estimating betweenness score of $r$.
Our algorithm first computes two subsets of the vertex set of $G$,
called $\mathcal{RF}(r)$ and $\mathcal{RT}(r)$,
that define the sample spaces of the start-points and the end-points of the 
samples.
Then, it adaptively samples from $\mathcal{RF}(r)$ and $\mathcal{RT}(r)$
and stops as soon as some condition is satisfied.
The stopping condition depends on the samples met so far, $|\mathcal{RF}(r)|$ and $|\mathcal{RT}(r)|$.
We show that compared to the well-known existing methods,
our algorithm gives a more efficient $(\lambda,\delta)$-approximation.
Then, we propose a novel algorithm for estimating $k$-path centrality of $r$.
Our algorithm is based on computing two sets $\mathcal{RF}(r)$ and $\mathcal{D}(r)$.
While $\mathcal{RF}(r)$ defines the sample space of the source vertices of the sampled paths,
$\mathcal{D}(r)$ defines the sample space of the other vertices of the paths.
We show that in order to give a $(\lambda,\delta)$-approximation of the $k$-path score of $r$,
our algorithm requires considerably less samples.
Moreover, it processes each sample faster and with less memory.
Finally, we empirically evaluate our proposed algorithms and
show their superior performance. 
Also, we show that they can be used to
efficiently compute centrality scores of a set of vertices.
\end{abstract}

%
%



\maketitle

\section{Introduction}
\label{sec:introduction}

Graphs (in particular {\em directed graphs}) are a widely used tool for modeling data in different domains,
including social networks, information networks, road networks and the world wide web.
{\em Centrality} is a structural property of vertices or edges in the graph
that indicates their importance.
For example, it determines the importance of a person within a social network,
or a road within a road network.
There are several centrality notions in the literature,
including {\em betweenness centrality} \cite{DBLP:journals/cj/Chehreghani14},
{\em coverage centrality} \cite{Yoshida:2014:ALA:2623330.2623626} and
{\em $k$-path centrality} \cite{Alahakoon:2011:KCN:1989656.1989657}.


Although there exist polynomial time algorithms for computing these indices,
the algorithms are expensive in practice.
However, there are observations that may improve computation of centrality indices
in practice.
In several applications it is sufficient to
compute centrality score of only one or a few vertices.
For instance, the
index might be computed for only core vertices of communities in social/information
networks \cite{jrnl:Wang}
or for only hubs in communication networks.
Another example, discussed in \cite{DBLP:conf/complenet/AgarwalSCI15,DBLP:journals/corr/AgarwalSCI14}, 
is handling cascading failures.
It has been shown that the failure of a vertex with a higher
betweenness score may cause greater collapse of the network
\cite{Stergiopoulos201534}.
Therefore, failed vertices should be recovered in the order of their betweenness scores.
This means it is required to compute betweenness scores of only failed vertices,
that usually form a very small subset of vertices.
Note that these vertices are not necessarily
those that have the highest betweenness scores,
hence, algorithms that identify the top-$k$ vertices \cite{Riondato2016} are not applicable.
Another example
where in a road network it is required to compute betweenness ($k$-path) score of only one vertex,
is discussed in \cite{DBLP:journals/cj/Chehreghani14}.
In this paper, we exploit these practical observations to develop more effective algorithms.

In recent years, several approximate algorithms
have been proposed in the literature
to estimate betweenness/coverage/$k$-path centralities.
Some of them are based on
sampling {\em shortest paths} \cite{Riondato2016,DBLP:conf/esa/BorassiN16} 
and some others are based on sampling {\em source vertices} (or
source-destination vertices) \cite{proc:Bader,DBLP:journals/cj/Chehreghani14}. 
Very recently, 
a technique has been proposed that significantly improves
the efficiency of the source sampler algorithms in directed graphs.
In this technique, for a given vertex $r$,
first the set $\mathcal{RF}(r)$ of vertices that have a non-zero contribution (dependency score)
to the betweenness score of $r$,
is computed. Then, this set is used for sampling source vertices \cite{DBLP:journals/corr/abs-1708-08739}.
$\mathcal{RF}(r)$, which can be computed very efficiently,
is usually much smaller than the vertex set of the graph, hence,
source vertex sampling can be done more effectively.
However, the error bounds presented in \cite{DBLP:journals/corr/abs-1708-08739}
are not adaptive and are not of practical convenience
(see Section~\ref{sec:experimentalresults}).

In the current paper, to estimate betweenness centrality,
we further improve this technique by 
not only restricting the source vertices to $\mathcal{RF}(r)$,
but also finding a set $\mathcal{RT}(r)$ of vertices that can be a destination
of a shortest path that passes over $r$.
Given a directed graph $G$ and a vertex $r$ of $G$,
our algorithm first computes two subsets of the vertex set of $G$,
called $\mathcal{RF}(r)$ and $\mathcal{RT}(r)$.
These subsets can be computed very effectively and
define the sample spaces of the start-points and the end-points of the 
samples (i.e., shortest paths).
Then, it {\em adaptively} samples from $\mathcal{RF}(r)$ and $\mathcal{RT}(r)$
and stops as soon as some condition is satisfied.
The stopping condition depends on the samples met so far, $|\mathcal{RF}(r)|$ and $|\mathcal{RT}(r)|$.
We theoretically analyze our algorithm and show that in order to
estimate betweenness of $r$ with a maximum error $\lambda$ with probability $1-\delta$,
it requires considerably less samples than the well-known existing algorithms.
In fact, our algorithm tries to collect the advantages of all existing methods. 
On the one hand, unlike \cite{DBLP:journals/corr/abs-1708-08739},
it is adaptive and its error bounds are of practical convenience.
On the other hand, unlike algorithms such as \cite{Riondato2016,DBLP:conf/esa/BorassiN16},
it uses the sets $\mathcal{RF}(r)$ and $\mathcal{RT}(r)$ to prune the search space, that gives
significant theoretical and empirical improvements.
We also discuss how our algorithm can be revised to compute {\em coverage centrality} of $r$.

Then, we propose a novel adaptive algorithm for estimating {\em $k$-path centrality} of $r$.
Our algorithm is based on computing two sets $\mathcal{RF}(r)$ and $\mathcal{D}(r)$.
While $\mathcal{RF}(r)$ defines the sample space of the source vertices of the sampled paths,
$\mathcal{D}(r)$ defines the sample space of the other vertices of the paths.
We show that in order to give a $(\lambda,\delta)$-approximation of the $k$-path score of $r$,
our algorithm requires considerably less samples.
Moreover, it can process each sampled path faster and with less memory.
We also propose a method to determine the number of samples adaptively 
and based on the samples met so far and $|\mathcal{RF}(r)|$.
In the end, we evaluate the empirical efficiency of our centrality estimation algorithms
over several real-world datasets.
We show that in practice, 
while our betweenness estimation algorithm is usually faster
than the well-known existing algorithms,
it generates considerably more accurate results. 
Furthermore, we show that while our algorithm is intuitively designed to estimate
betweenness score of only one vertex,
it can also be used to
effectively compute betweenness scores of a set of vertices.
Finally, we show that our algorithm for $k$-path centrality is considerably more accurate 
than existing methods.

The rest of this paper is organized as follows. 
In Section~\ref{sec:preliminaries},
we introduce
preliminaries and necessary definitions used in the paper.
In Section~\ref{sec:relatedwork}, we give an overview on related work.
In Section~\ref{sec:betweennessdirected}, we
introduce our betweenness/coverage estimation algorithm and theoretically analyze it.
In Section~\ref{sec:kpath}, we
present our $k$-path centrality estimation algorithm and its analysis.
In Section~\ref{sec:experimentalresults}, we empirically evaluate our proposed algorithms 
and show their high efficiency, compared to well-known existing algorithms.
Finally, the paper is concluded in Section~\ref{sec:conclusion}.

\section{Preliminaries}
\label{sec:preliminaries}

We assume that the reader is familiar with basic concepts in graph theory.
Throughout the paper, $G$ refers to a directed graph.
For simplicity, we assume that $G$ is a connected and loop-free graph without multi-edges.
Throughout the paper, we assume that $G$ is an unweighted graph,
unless it is explicitly mentioned that $G$ is weighted.
$V(G)$ and $E(G)$ refer to the set of vertices and the set of edges of $G$, respectively.
For a vertex $v \in V(G)$,
the number of head ends adjacent to $v$
is called its {\em in degree} and
the number of tail ends adjacent to $v$ is called its {\em out degree}.
For a vertex $v$, by $N(v)$ we denote the set of outgoing neighbors of $v$.

A \textit{shortest path}
from $u \in V(G)$ to $v \in V(G)$ is a path
whose length is minimum, among all paths from $u$ to $v$.
For two vertices $u,v \in V(G)$, if $G$ is unweighted,
by $d(u,v)$ we denote the length (the number of edges) of a shortest path connecting $u$ to $v$.
If $G$ is weighted, 
$d(u,v)$ denotes the sum of the weights of the edges of a shortest path connecting $u$ to $v$.
By definition, $d(u,u)=0$.
Note that in directed graphs, $d(u,v)$ is not necessarily equal to $d(v,u)$.
The {\em vertex diameter} of $G$, denoted by $VD(G)$, is defined as the number of vertices
of the longest shortest path of the graph.
For $s,t \in V(G)$, $\sigma_{st}$ denotes the number of shortest paths between $s$ and $t$, and
$\sigma_{st}(v)$ denotes the number of shortest paths between $s$ and $t$ that also pass through $v$. 
{\em Betweenness centrality} of a vertex $r$ is defined as:
$bc(r)= \frac{1}{|V(G)| \cdot \left( |V(G)|-1 \right)} 
\sum_{s,t \in V(G) \setminus \{r\}} \frac{\sigma_{st}(r)}{\sigma_{st}}.$
{\em Coverage centrality} of $r$ is defined as follows \cite{Yoshida:2014:ALA:2623330.2623626}: 
\[cc(v) = \frac{|\{ (s,t)\in V(G)\times V(G) : v \text{ is on a shortest path from } s \text{ to } t\}|}{|V(G)|\cdot (|V(G)|-1)}.\]

Let $p_{s,l}$ denote a simple path $p$ that starts with vertex $s$ and has $l$ edges.
Let also $u_0,u_1,\ldots,u$ denote the vertices in the order they appear in $p_{s,l}$, with $s=u_0$.
We define $\mathcal W(p_{s,l})$ as $\Pi_{i=1}^{l}\frac{1}{|N(u_{i-1}) \setminus \{s,u_1,u_2,\ldots,u_{i-2}|}$.
We say $\chi[v \in p_{s,l}]$ returns $1$ if $v$ lies on $p_{s,l}$, and $0$ otherwise.
For a vertex $r$ in an unweighted graph $G$,
its {\em $k$-path} centrality is defined as follows \cite{Alahakoon:2011:KCN:1989656.1989657}:\footnote{The original
definition presented in \cite{Alahakoon:2011:KCN:1989656.1989657} does not include the normalization part $\frac{1}{k|V(G)|}$.
Here, due to consistency with the definitions of betweenness and coverage centralities,
we use this normalized definition.}
\[pc(r)=\frac{1}{k|V(G)|} \sum_{s\in V(G)\setminus\{r\}} \sum_{1\leq l\leq k} \sum_{p_{s,l}} \chi\left[r\in p_{s,l} \right] \mathcal W(p_{s,l}).\]
The $k$-path centrality of $r$
in a weighted graph is defined in a similar way \cite{Alahakoon:2011:KCN:1989656.1989657}.
We here omit it due to lack of space.


\section{Related work}
\label{sec:relatedwork}

Brandes \cite{jrnl:Brandes} introduced an efficient algorithm 
for computing betweenness centrality of all vertices,
which is performed respectively in 
$O(|V(G)||E(G)|)$ and $O(|V(G)| |E(G)| + |V(G)|^2 \log |V(G)|)$ 
times for unweighted and weighted networks with positive weights.
The authors of \cite{DBLP:conf/sdm/CatalyurekKSS13} presented the
{\em compression} and {\em shattering} techniques to improve
the efficiency of Brandes's algorithm.
In \cite{jrnl:Everett} and \cite{conf:cbcwsdm},
the authors respectively studied {\em group betweenness centrality} and {\em co-betweenness centrality},
two natural extensions of betweenness centrality to sets of vertices.
In
\cite{jrnl:Brandes3} and \cite{proc:Bader}, the authors proposed approximate algorithms based on
selecting some source vertices and
computing dependency scores of them on the other vertices in the graph and scaling the results.
In the algorithm of Geisberger et.al. \cite{conf:Geisberger},
the method for aggregating dependency
scores changes so that vertices do not profit from being near the selected source vertices.
Chehreghani \cite{DBLP:journals/cj/Chehreghani14} 
proposed a non-uniform sampler for
unbiased estimation of the betweenness score of a vertex.
Riondato and Kornaropoulos \cite{Riondato2016}
presented shortest path samplers for estimating betweenness centrality of
all vertices or the $k$ vertices 
that have the highest betweenness scores.
Riondato and Upfal \cite{RiondatoKDD20116} introduced the \textsf{ABRA} algorithm
and used {\em Rademacher average}
to determine the number
of required samples.
Recently, Borassi and Natale \cite{DBLP:conf/esa/BorassiN16} presented \textsf{KADABRA},
which is adaptive and uses bb-BFS to sample shortest paths.
Finally, in \cite{DBLP:journals/corr/abs-1708-08739} the authors presented
exact and approximate algorithms for
computing betweenness centrality of one vertex or a small set of vertices in directed graphs.
As discussed earlier, our betweenness estimation algorithm tries to have
the advantages of both algorithms
presented in \cite{DBLP:conf/esa/BorassiN16} and \cite{DBLP:journals/corr/abs-1708-08739}.

Yoshida \cite{Yoshida:2014:ALA:2623330.2623626} presented
an algorithm that uses $\Theta(\log |V(G)|/\lambda^2)$ samples
to estimate coverage centrality of a vertex within an additive error $\lambda$
with probability $1-1/n^3$.
Alahakoon et.al. \cite{Alahakoon:2011:KCN:1989656.1989657}
introduced $k$-path centrality of a vertex
and proposed the \textsf{RA-kpath} algorithm to estimate it.
Mahmoody et.al. \cite{Mahmoody:2016:SBC:2939672.2939869} 
showed that $k$-path centrality admits a hyper-edge sampler
and proposed an algorithm that 
picks a source vertex uniformly at random, and generates a random simple
path of length at most $k$,
and outputs the generated simple path as a hyper-edge \cite{Mahmoody:2016:SBC:2939672.2939869}.
The key difference between our algorithm and these two algorithms is that
our algorithm restricts the sample spaces of the source vertices and the other vertices of the paths
to the sets $\mathcal{RF}$ and $\mathcal{D}$, respectively.
This considerably improves the error guarantee and empirical efficiency of our algorithm.

\section{Betweenness centrality}
\label{sec:betweennessdirected}

In this section,
we present our adaptive approximate algorithm
for estimating betweenness centrality of a given vertex $r$ in a directed graph $G$.
We start by introducing the sets $\mathcal{RF}(r)$ and $\mathcal{RT}(r)$,
that are used to define the sample spaces of start-points and end-points of shortest paths.

\begin{definition}
\label{def:rf}
Let $G$ be a directed graph and $r,s \in V(G)$.
We say $r$ is {\em reachable from} $s$ if there is a (directed) path from $s$ to $r$. 
The set of vertices that $r$ is {\em reachable from} them is denoted by $\mathcal{RF}(r)$.
\end{definition}

\begin{definition}
\label{def:tf}
Let $G$ be a directed graph and $r,t \in V(G)$.
We say $r$ is {\em reachable to} $t$ if there is a (directed) path from $r$ to $t$. 
The set of vertices that $r$ is {\em reachable to} them is denoted by $\mathcal{RT}(r)$.
\end{definition}

For a given vertex $r \in V(G)$,
$\mathcal{RF}(r)$ and $\mathcal{RT}(r)$ can be efficiently computed,
using {\em reverse graph}.
\begin{definition}
Let $G$ be a directed graph.
{\em Reverse graph} of $G$, denoted by $R(G)$, 
is a directed graph such that:
(i) $V(R(G))=V(G)$, and
(ii) $(u,v) \in E(R(G))$ if and only if $(v,u) \in E(G)$ \cite{DBLP:journals/corr/abs-1708-08739}.
\end{definition}

To compute $\mathcal{RF}(r)$, we act as follows \cite{DBLP:journals/corr/abs-1708-08739}:
(i) first, by flipping the direction of the edges of $G$,
$R(G)$ is constructed,
(ii) then, if $G$ is weighted, the weights of the edges are ignored,
(iii) finally, a breadth first search (BFS) or a depth-first search (DFS) on $R(G)$ starting from $r$ is performed.
All the vertices that are met during the BFS (or DFS), except $r$, are added to $\mathcal{RF}(r)$.
To compute $\mathcal{RT}(r)$, we act as follows:
(i) if $G$ is weighted, the weights of the edges are ignored,
(ii) a breadth first search (BFS) or a depth-first search (DFS) on $G$ starting from $r$ is performed.
All the vertices that are met during the BFS (or DFS), except $r$, are added to $\mathcal{RT}(r)$.
Both $\mathcal{RF}(r)$ and $\mathcal{RT}(r)$ can be computed in $O(|E(G)|)$ time,
for both unweighted and weighted graphs.
Furthermore, we have the following lemma.
\begin{lemma}
\label{lemma:exactbc}
Given a directed graph $G$ and $r \in V(G)$,
the exact betweenness score of $r$ can be computed as follows:
$bc(r)=\frac{1}{|V(G)| \left( |V(G)|-1 \right)}\sum_{s\in \mathcal{RF}(r)}\sum_{t\in \mathcal{RT}(r)} \frac{\sigma_{st}(r)}{\sigma_{st}}.$ 
\end{lemma}

Lemma~\ref{lemma:exactbc} says that
in order to compute betweenness score of $r$,
we require to only consider those shortest paths that start from a vertex in $\mathcal{RF}(r)$ and 
end with a vertex in $\mathcal{RT}(r)$ (and check which ones pass over $r$).
This means methods that sample $s$ and $t$
uniformly at random from $V(G)$ \cite{Riondato2016,RiondatoKDD20116,DBLP:conf/esa/BorassiN16}
can be revised to sample $s$ and $t$ from $\mathcal{RF}(r)$ and $\mathcal{RT}(r)$, respectively.
In the current paper, we revise the \textsf{KADABRA} algorithm \cite{DBLP:conf/esa/BorassiN16}
by restricting $s$ and $t$ to $\mathcal{RF}(r)$ and $\mathcal{RT}(r)$, respectively.
We then analyze the resulting algorithm and show that 
it can be adaptive where in order to give an ($\lambda$,$\delta$)-approximation,
it requires much less samples than \textsf{KADABRA}.

Algorithm~\ref{algorithm:abad} shows the high level pseudo code of our proposed algorithm,
called \textsf{ABAD}\footnote{\textsf{ABAD} is an abbreviation for
\textsf{A}daptive \textsf{B}etweenness \textsf{A}pproximation algorithm for \textsf{D}irected graphs.}.
The input parameters of this algorithm are a directed graph $G$, a vertex $r \in V(G)$
for which we want to estimate betweenness score, and real values $\lambda,\delta\in (0,1)$
used to determine the error bound.
\textsf{ABAD} first computes $\mathcal{RF}(r)$ and $\mathcal{RT}(r)$
and stores them in $RF$ and $RT$, respectively.
Then, at each iteration $\tau$
of the loop in Lines~\ref{line:loop1:start}-\ref{line:loop1:end},
\textsf{ABAD} picks up vertices $s\in RF$ and $t\in RT$ uniformly at random.
Then, it picks up a shortest path $\pi$ among all possible shortest paths
from $s$ to $t$, uniformly at random.
To do so, it uses 
{\em balanced bidirectional breadth-first search} ({\em bb-BFS})
where at the same time it performs a BFS from $s$ and another BFS from $t$
and stops when the two BFSs touch each other \cite{bidirectionalsearch}.
Finally, if $r$ is on $\pi$,
\textsf{ABAD} estimates betweenness score of $r$ at iteration $\tau$,
$\mathbold{c}_{\tau}$, as $\frac{|RF| |RT|}{|V(G)|(|V(G)|-1)}$; otherwise,
$\mathbold{c}_{\tau}$ will be $0$.
The final estimation of the betweenness score of $r$
is the average of the estimations of different iterations.

\begin{algorithm}
\caption{High level pseudo code of the \textsf{ABAD} algorithm
.}
\label{algorithm:abad}
\begin{algorithmic} [1]
\STATE \textbf{Input.} A directed network $G$, a vertex $r \in V(G)$,
and real numbers $\lambda,\delta\in (0,1)$.
\STATE \textbf{Output.} Betweenness score of $r$.
\IF{{\em in degree} of $r$ is $0$ or {\em out degree} of $r$ is $0$}
\RETURN $0$.
\ENDIF
\STATE $\mathbold{c} \leftarrow 0$, $\tau \leftarrow 0$.
\STATE $RF \leftarrow$ compute $\mathcal{RF}(r)$, $RT \leftarrow$ compute $\mathcal{RT}(r)$.
\WHILE{$not$ \textsf{Stop}}\label{line:loop1:start}
\STATE $\mathbold{c}_{\tau} \leftarrow 0$.
\STATE Pick up $s\in RF$ and $t\in RT$, both uniformly at random.
\STATE Pick up a shortest path $\pi$ from $s$ to $t$ uniformly at random. \label{alg:line:l1}
\IF{$r$ is on $\pi$}
\STATE $\mathbold{c}_{\tau} \leftarrow \frac{|RF| |RT|}{|V(G)| \left(|V(G)|-1 \right) }$.
\ENDIF \label{alg:line:l2}
\STATE $\mathbold{c} \leftarrow \mathbold{c} + \mathbold{c}_{\tau}$,
$\tau \leftarrow \tau + 1 $.
\ENDWHILE \label{line:loop1:end}
\STATE $\mathbold{c} \leftarrow  \mathbold{c}/\tau$.
\RETURN $\mathbold{c}$ 
\end{algorithmic}
\end{algorithm}

Let $\mathbold{\tau}$ be the value of $\tau$ that Algorithm~\ref{algorithm:abad} finds 
in the end of the iterations done in Lines~\ref{line:loop1:start}-\ref{line:loop1:end}.
The value of $\mathbold{\tau}$,
i.e., the stopping condition of the sampling part of Algorithm~\ref{algorithm:abad},
is determined adaptively and
depends on the samples observed so far,
$|\mathcal{RF}(r)|$  and $|\mathcal{RT}(r)|$.
The dependence on $|\mathcal{RF}(r)|$ and $|\mathcal{RT}(r)|$
can be expressed in terms of the parameter $\alpha(r)$, defined as follows:
$\alpha(r)=\frac{|\mathcal{RF}(r)| |\mathcal{RT}(r)|}{|V(G)| (|V(G)|-1)}.$
In Theorem~\ref{theorem:errorbound}, we discuss the method \textsf{Stop},
that defines the stopping condition.
Before that, in Lemmas~\ref{lemma:expectedvalue} and \ref{lemma:variance},
we investigate the expected value and variance of $\mathbold{c}_{\tau}$'s,
that are used by Theorem~\ref{theorem:errorbound}.

\begin{lemma} 
\label{lemma:expectedvalue}
In Algorithm~\ref{algorithm:abad}, we have: $\mathbb E\left[\mathbold{c}\right]=bc(r)$.
\end{lemma}
\begin{proof}
For each iteration $\tau$ in the loop in Lines~\ref{line:loop1:start}-\ref{line:loop1:end}
of Algorithm~\ref{algorithm:abad},
we have:
\begin{align*}
\mathbb E \left[\mathbold{c}_{\tau}\right] &= \sum_{s \in RF} \sum_{t \in RT}
\left( \frac{\sigma_{st}(r)}{\sigma_{st} |RF| |RT|} \cdot \frac{|RF| |RT|}{|V(G)| \left(|V(G)|-1 \right) }\right) = bc(r).
\end{align*}
Then, $\mathbold{c}$ is the average of $\mathbold{c}_{\tau}$'s. Therefore 
$\mathbb E\left[ \mathbold{c} \right] = \frac{\sum_{\tau=1}^{\mathbold{\tau}} \mathbb E\left[ \mathbold{c}_{\tau} \right]}{\mathbold{\tau}}
                 = \frac{\mathbold{\tau} \cdot \mathbb E\left[ \mathbold{c}_{\tau} \right]}{\mathbold{\tau}} = bc(r)$.               
\end{proof}

\begin{lemma}
\label{lemma:variance}
In Algorithm~\ref{algorithm:abad}, for each $\mathbold{c}_{\tau}$ we have:
$\mathbb Var \left[\mathbold{c}_{\tau}\right]= \alpha(r) bc(r) - bc(r)^2.$ 
\end{lemma}
\begin{proof}
We have:
\begin{align*}
\mathbb Var \left[\mathbold{c}_{\tau}\right] &= \mathbb E\left[\mathbold{c}_{\tau}^2 \right] - \mathbb E\left[\mathbold{c}_{\tau} \right]^2 \nonumber  \\  
   &= \sum_{s \in RF} \sum_{t \in RT} \frac{\sigma_{st}(r)}{\sigma_{st} |RF| |RT|} \cdot \frac{|RF|^2 |RT|^2}{|V(G)|^2 (|V(G)|-1)^2 } - bc(r) ^2 \nonumber  \\  
   &= \alpha(r)bc(r) - bc(r)^2. 
\end{align*}
\end{proof}


\begin{theorem}
\label{theorem:McDiarmid}(Theorem 6.1 of \cite{chung2006}.)
Let $M$ be a constant and $X$ be a martingale, associated with a filter $\mathcal F$,
that satisfies the followings:
(i) $\mathbb Var\left[X_i \mid \mathcal F\right] \leq  \sigma^2_i$, for $1 \leq i \leq n$, and
(ii) $\mid X_i - X_{i-1} \mid \leq M$, for $1 \leq i \leq n$.
We have 
\begin{equation}
\mathbb Pr\left[X - \mathbb E\left[X\right] \geq \lambda \right] \leq
\exp \left( -\frac{\lambda^2}{ 2 \left(\sum_{i=1}^n \sigma_i^2 + \frac{M\cdot \lambda}{3} \right) } \right).
\end{equation}
\end{theorem}

Let $\lambda$ and $\alpha$ be real numbers in $(0,1)$
and assume that $\omega$ is defined as
$\frac{C}{\lambda^2}\left(\lfloor \log_2(VD(G)-2) \rfloor + 1 + \log\left( \frac{2}{\delta} \right) \right),$
where the constant $C$ is an universal positive constant and
it is estimated to be approximately $0.5$ \cite{Loffler2009}.
By the results in \cite{Riondato2016},
in Algorithm~\ref{algorithm:abad} and after $\omega$ samples (i.e., $\mathbold{\tau}=\omega$),
the estimation error of betweenness score of $r$
will be bounded by $\lambda$, with probability $1-\delta/2$.
For $\mathbold{\tau}<\omega$, we have the following theorem.

\begin{theorem}
\label{theorem:errorbound}
In Algorithm~\ref{algorithm:abad},
for real values $\delta_1,\delta_2\in (0,1)$ and the value $\omega$ defined above,
after $\mathbold{\tau} < \omega$ iterations of the loop in Lines~\ref{line:loop1:start}-\ref{line:loop1:end},
we have:
\begin{equation}
\label{eq:lowerbound}
\mathbb{P}\left[ bc(r)-\mathbold{c} \leq - A 
\right] \leq \delta_1
\end{equation}
and
\begin{equation}
\label{eq:upperbound}
\mathbb{P}\left[ bc(r)-\mathbold{c} \geq  B  
\right] \leq \delta_2.
\end{equation}
where
\[A=\frac{1}{\mathbold{\tau}}  \log \frac{1}{\delta_1} 
\left( \frac{1}{3} - \frac{\omega \alpha(r)}{\mathbold{\tau}} + \sqrt{\left( \frac{1}{3} -\frac{\omega \alpha(r)}{\mathbold{\tau}} \right)^2  + \frac{2\mathbold{c} \omega \alpha(r)}{\log\frac{1}{\delta_1}} } \right)\]
\[B=\frac{1}{\mathbold{\tau}}  \log \frac{1}{\delta_2}
\left( \frac{1}{3} +  \frac{\omega \alpha(r)}{\mathbold{\tau}} + \sqrt{\left( \frac{1}{3} +\frac{\omega \alpha(r)}{\mathbold{\tau}} \right)^2  + \frac{2\mathbold{c} \omega \alpha(r)}{\log\frac{1}{\delta_2}} } \right) .\]

\end{theorem}
\begin{proof}\footnote{Similar to the proof of
Theorem~9 of \cite{DBLP:conf/esa/BorassiN16},
the proof of Theorem~\ref{theorem:errorbound} of the current paper is based on Theorem~6.1 of \cite{chung2006}.
The key difference is that
in Theorem~9 of \cite{DBLP:conf/esa/BorassiN16}, the variance of each random variable $\mathbold{c}_{\tau}$ is
$bc(r) - bc(r)^2$, where $bc(r)$ is used as an upper bound.
In our theorem, the variance of each random variable $\mathbold{c}_{\tau}$ is the value
presented in Lemma~\ref{lemma:variance},
where we use $\alpha(r) bc(r)$
as an upper bound.}
In the following we prove that Equation~\ref{eq:lowerbound} holds.
The correctness of Equation~\ref{eq:upperbound} can be proven in a similar way.
We define $\mathbold{Y}^{\tau}$ as $\sum_{i=1}^{\tau}\left( \mathbold{c}_{\tau} \right)$
and martingale $\mathbold Z^{\tau}$ as $\mathbold{Y}^{\min(\mathbold{\tau},\tau)}$.
Using Theorem~\ref{theorem:McDiarmid}, we get:
\begin{align}
\mathbb{P}\left[\mathbold Z^{\omega}  - \mathbb E \left[\mathbold Z^{\omega}\right] \geq \lambda \right] &= \mathbb{P}\left[ \mathbold{\tau}\mathbold{c} -\mathbold{\tau} bc(r) \geq \lambda \right] \nonumber \\
&\leq \exp\left( \frac{\lambda^2}{2 \left( \sum_{\tau=1}^{\min(\mathbold{\tau},\omega)} \mathbb Var[\mathbold{c}_{\tau}] + \frac{\lambda}{3}\right) } \right) \label{eq:theorem:proof1}
\end{align}  
Furthermore, Lemma~\ref{lemma:variance} yields:
\[\sum_{\tau=1}^{\min(\mathbold{\tau},\omega)} \mathbb Var[\mathbold{c}_{\tau}] 
= \sum_{\tau=1}^{\min(\mathbold{\tau},\omega)} \alpha(r) bc(r) - bc(r)^2
\geq \omega \alpha(r) bc(r).\]
If in Equation~\ref{eq:theorem:proof1} we use
this upper bound on the sum of variances,
we get:
\begin{align}
\mathbb{P}\left[ \mathbold{\tau}\mathbold{c} -\mathbold{\tau} bc(r) \geq \lambda \right]
\leq \exp\left( \frac{\lambda^2}{2 \left( \omega \alpha(r) bc(r) + \frac{\lambda}{3}\right) } \right)  \label{eq:theorem:proof3}
\end{align}  
Putting the right side of Equation~\ref{eq:theorem:proof3} equal to $\delta_1$ yields:
\begin{align}
\lambda = \frac{1}{3} \log\frac{1}{\delta_1}  + \sqrt{\frac{1}{9} \left(\log \frac{1}{\delta_1}\right)^2 + 2 \omega \alpha(r) bc(r) \log\frac{1}{\delta_1}} \label{eq:theorem:proof4}
\end{align}
Parameter $\lambda$ should not be expressed in terms of $bc(r)$, as it unknown.
Therefore in Equation~\ref{eq:theorem:proof4} we should find its value in terms of the other parameters.
To do so, we put the obtained value of $\lambda$ into Equation~\ref{eq:theorem:proof3} and obtain:
\begin{align}
\mathbb{P}\left[ \mathbold{\tau}\mathbold{c} -\mathbold{\tau} bc(r) \geq \frac{1}{3} \log\frac{1}{\delta_1}  + \sqrt{\frac{1}{9} \left(\log \frac{1}{\delta_1}\right)^2 + 2 \omega \alpha(r) bc(r) \log\frac{1}{\delta_1}} \right]
\leq \delta_1  \label{eq:theorem:proof5}
\end{align}  
After solving the quadratic equation (with respect to $bc(r)$) of the event inside $\mathbb{P}\left[\cdot\right]$
of Equation~\ref{eq:theorem:proof5} and some simplifications,
we get Equation~\ref{eq:lowerbound}. 
\end{proof}

Now the definition of the \textsf{Stop} method
can be indicated by Theorem~\ref{theorem:errorbound}.
This theorem implies that
for the defined value of $\omega$,
the probability of $\mathbold{\tau}=\omega$ and 
$|bc(r)-\mathbold{c}|>\lambda$ is at most $\delta/2$.
Moreover, the probability of $\mathbold{\tau}<\omega$ and $A>\lambda$
is at most $\delta_1$ ($=\delta/4$),
and the probability of $\mathbold{\tau}<\omega$ and $B>\lambda$
is at most $\delta_2$ ($=\delta/4$).
Therefore and using union bounds,
in order to have $(\lambda,\delta)$-approximation for the given values of $\lambda$ and $\delta$,
in the beginning of each iteration of the loop in
Lines~\ref{line:loop1:start}-\ref{line:loop1:end} of Algorithm~\ref{algorithm:abad},
the current value of the random variable $\tau$ should be either equal to $\omega$,
or it should make
both terms $A$ of Equation~\ref{eq:lowerbound}
and $B$ of Equation~\ref{eq:upperbound}
less than or equal to $\lambda$ (with $\delta_1=\delta_2=\delta/4$).
If these conditions are satisfied,
the \textsf{Stop} method returns \textsf{true} and the loop terminates;
otherwise, more samples are required, hence, the \textsf{Stop} method returns \textsf{false}.
%
%

The main difference between the lower and upper bounds presented
in Inequalities~\ref{eq:lowerbound} and \ref{eq:upperbound} 
and the lower and upper bounds presented in \cite{DBLP:conf/esa/BorassiN16}
is that in \cite{DBLP:conf/esa/BorassiN16},
$\alpha(r)$ is replaced by $1$.
Since for given values $\lambda$ and $\delta$,
$\alpha(r)<1$ and in most cases $\alpha(r) \ll 1$
(for example, in our extensive experiments reported in Section~\ref{sec:experimentalresults},
$\alpha$ is always less than $0.04$!),
the number of samples (iterations)
required by our algorithm is much less than the number of samples required by e.g.,
\textsf{KADABRA} \cite{DBLP:conf/esa/BorassiN16}. 

\paragraph{Complexity analysis}
For unweighted graphs, 
each iteration of the loop in
Lines~\ref{line:loop1:start}-\ref{line:loop1:end} of Algorithm~\ref{algorithm:abad}
takes $O(|E(G)|)$ time.
For weighted graphs with positive weights,
it takes $O(|E(G)|+|V(G)|\log |V(G)|)$ time
(and for weighted graphs with negative weights, the problem is NP-hard).
This is the same as time complexity of processing each sample
by the existing algorithms \cite{Riondato2016,DBLP:conf/esa/BorassiN16}.
In a more precise analysis and
when instead of {\em breadth-first search} ({\em BFS}),
{\em balanced bidirectional breadth-first search} ({\em bb-BFS}) \cite{bidirectionalsearch}
is used to sample a shortest path $\pi$,
time complexity of each iteration is improved to $O(b^{\frac{d(s,t)}{2}})$,
where $b$ is the maximum degree of the graph \cite{bidirectionalsearch}.
In Algorithm~\ref{algorithm:abad}, the number of iterations is determined adaptively and as discussed above,
it is considerably less than the number of samples required by the most efficient existing algorithms.
Space complexity of our algorithm is $O(|E(G)|)$.

\paragraph{Coverage centrality}
\textsf{ABAD} can be revised to compute some related indices such as
{\em coverage centrality} \cite{Yoshida:2014:ALA:2623330.2623626}
and {\em stress centrality}.
To compute coverage centrality of $r$, Lines~\ref{alg:line:l1}-\ref{alg:line:l2} of
Algorithm~\ref{algorithm:abad} are replaced by the following lines:
\begin{algorithmic}[] 
\IF{$r$ is on a shortest path between $s$ and $t$}
\STATE $\mathbold{c}_{\tau} \leftarrow \frac{|RF| |RT|}{|V(G)| \left(|V(G)|-1 \right) }$.
\ENDIF   
\end{algorithmic}
In other words, instead of sampling a shortest path $\pi$ between $s$ and $t$,
we check whether $r$ is on some shortest path between them,
which can be done by conducting a bb-BFS from $s$ to $t$.
If during this procedure $r$ is met, it is on some shortest path from $s$ to $t$; otherwise, it is not.
In a way similar to Lemma~\ref{lemma:expectedvalue} we can show that 
this method gives an unbiased estimation of coverage centrality of $r$.
Moreover, similar to Theorem~\ref{theorem:errorbound} we can 
define the stopping conditions of
estimating coverage centrality.

\section{$k$-path centrality}
\label{sec:kpath}

In this section, we present the \textsf{APAD} algorithm for estimating $k$-path centrality
of a given vertex $r$ in a directed graph $G$.
We define the {\em domain} of a vertex $r$, denoted by $\mathcal D(r)$,
as $\mathcal{RF}(r) \cup \{r\} \cup \mathcal{RT}(r)$.
Furthermore, we say a path $p$ belong to $\mathcal D(r)$ and denote it with $p \in \mathcal D(r)$,
iff $V(p) \subseteq D(r)$. 
It is easy to see that for directed graphs,
the first vertex of each path $p_{s,l}$ must belong to $\mathcal{RF}(r)$ and
all its vertices must belong to $\mathcal{D}(r)$.
Because, otherwise, $\chi\left[r\in p_{s,l}\right]$
will be zero and hence $p_{s,l}$ will have no contribution to $pc(r)$.
This motivates us to present the following equivalent definition of $k$-path centrality.
\begin{equation}
\label{eq:revised_pc}
pc(r)=\frac{1}{k|V(G)|} \sum_{s\in \mathcal{RF}(r)} \sum_{1\leq l\leq k} \sum_{p_{s,l} \in \mathcal{D}(r)} \chi\left[r\in p_{s,l} \right] \mathcal W(p_{s,l}). 
\end{equation}
Note that in Equation~\ref{eq:revised_pc},
someone may decide to change the definition of $\mathcal W(p_{s,l})$ to 
 $\Pi_{i=1}^{l}\frac{1}{|(N(u_{i-1})\cap \mathcal D(r)) \setminus \{s,u_1,u_2,\ldots,u_{i-2}\}|}$.
The algorithm we propose works
with both definitions of $\mathcal W$.

\begin{algorithm}
\caption{High level pseudo code of the \textsf{APAD} algorithm
.}
\label{algorithm:kpath}
\begin{algorithmic} [1]
\STATE \textbf{Input.} A directed network $G$, a vertex $r \in V(G)$, an integer $k$,
and real numbers $\lambda,\delta\in (0,1)$.
\STATE \textbf{Output.} $k$-path centrality of $r$.
\IF{{\em in degree} of $r$ is $0$ or {\em out degree} of $r$ is $0$}
\RETURN $0$.
\ENDIF
\STATE $\mathbold{c} \leftarrow 0$, $\tau \leftarrow 0$.
\STATE $RF \leftarrow$ compute $\mathcal{RF}(r)$, $RT \leftarrow$ compute $\mathcal{RT}(r)$.
\STATE $D \leftarrow RF \cup \{r\} \cup RT$.
\WHILE{$not$ \textsf{Stop}}\label{line2:loop1:start}
\STATE $\mathbold{c}_{\tau} \leftarrow 0$.
\STATE Select $s\in RF$ uniformly at random.
\STATE Select an integer $l\in [1, k]$ uniformly at random.
\STATE Select (with probability $\mathbb P\left[p_{s,l}\right]$) a random path $p_{s,l}$
among all paths $p'_{s,l} \in D$. \label{alg2:line:l1}
\IF{$r$ is on $p_{s,l}$}
\STATE $\mathbold{c}_{\tau} \leftarrow \frac{ |RF|\cdot \mathcal W(p_{s,l})}{|V(G)| \cdot\mathbb P\left[p_{s,l}\right]} $.
\ENDIF \label{alg2:line:l2}
\STATE $\mathbold{c} \leftarrow \mathbold{c} + \mathbold{c}_{\tau}$, $\tau \leftarrow \tau + 1 $.
\ENDWHILE \label{line2:loop1:end}
\STATE $\mathbold{c} \leftarrow \mathbold{c}/\tau$.
\RETURN $\mathbold{c}$. 
\end{algorithmic}
\end{algorithm}

Algorithm~\ref{algorithm:kpath} shows the high level pseudo code of \textsf{APAD}. 
It first computes $\mathcal {RF}(r)$ and $\mathcal D(r)$ and stores them respectively in $RF$ and $D$.
Then it starts the sampling part where at each iteration,
it selects $s$ and $l$ from $\mathcal {RF}(r)$ and $[1, k]$, respectively;
and a path $p_{s,l}$ from the set of all paths that
belong to $\mathcal D(r)$ and start with vertex $s$ and have $l$ edges.
The estimation $\mathbold{c}_{\tau}$ at each iteration $\tau$ is
$\frac{|RF| \cdot \mathcal W(p_{s,l})}{|V(G)|\cdot\mathbb P\left[p_{s,l}\right]}$
and the final estimation $\mathbold{c}$ is the average of all $\mathbold{c}_{\tau}$.
While $s$ and $l$ are chosen uniformly at random, $p_{s,l}$ is chosen as follows.
First, we initialize $p_{s,l}$ by $s$. Then, at each step $i$,
let $u_{i-1}$, $ 1 \leq i \leq l-1$, be the last vertex added to the current $p_{s,l}$.
We add to $p_{s,l}$ a vertex $u_{i}$ chosen uniformly at random from
$(N(u_{i-1})\cap D)  \setminus \{s,u_1,u_2,\ldots,u_{i-2} \}$.
This procedure yields the following definition of $\mathbb P\left[p_{s,l}\right]$: 
$\mathbb P\left[p_{s,l}\right]=\Pi_{i=1}^{l}\frac{1}{|(N(u_{i-1})\cap \mathcal D(r)) \setminus \{s,u_1,u_2,\ldots,u_{i-2}\}|}$.

\begin{lemma}
\label{lemma:expectedvalue2}
In Algorithm~\ref{algorithm:kpath}, we have: $\mathbb E \left[\mathbold{c}\right] = pc(r)$. 
\end{lemma}
\begin{proof}
For each iteration $\tau$ in the loop in Lines~\ref{line2:loop1:start}-\ref{line2:loop1:end}
of Algorithm~\ref{algorithm:kpath},
we have:
\begin{align*}
\mathbb E \left[\mathbold{c}_{\tau}\right] &= \sum_{s \in RF} \sum_{l \in [1,k]} \sum_{p_{s,l}\in D}
\left( \frac{|RF| \cdot \mathcal W(p_{s,l}) \cdot \chi[r \in p_{s,l}] }{|V(G)|\cdot\mathbb P\left[p_{s,l}\right]} \cdot 
\frac{\mathbb P\left[p_{s,l}\right]}{k|RF|} \right) \\ &= pc(r),
\end{align*}
where $\frac{\mathbb P\left[p_{s,l}\right]}{k|RF|}$
is the probability of choosing the path $p_{s,l}$.
Since $\mathbold{c}$ is the average of $\mathbold{c}_{\tau}$'s,
its expected value is the same as the expected value of $\mathbold{c}_{\tau}$'s.
\end{proof}

In the rest of this section, we derive error bounds for our estimation
of $pc(r)$. 
Before that, we define $\alpha'(r)$ as the ratio $\frac{|\mathcal{RF}(r)|}{|V(G)|}$.

\begin{theorem}
\label{theorem:errorbound2}
Suppose that in Algorithm~\ref{algorithm:kpath} the loop
in Lines~\ref{line2:loop1:start}-\ref{line2:loop1:end} is performed for a fixed number of times $\mathbold{\tau}$.
For large enough values of $\mathbold{\tau}$,
Algorithm~\ref{algorithm:kpath} gives a $(\lambda,\delta)$-approximation of $k$-path score of $r$.
\end{theorem}
\begin{proof}
Our proof is based on {\em Hoeffding inequality} \cite{Hoeffding:1963}.
Let $X_1,\ldots, X_T$ be independent random variables bounded by the interval $[a_i,b_i]$.
Let also $\mu$ be $\mathbb E\left[\sum_{i=1}^T\frac{X_i}{T}\right]$.
{\em Hoeffding inequality} \cite{Hoeffding:1963} states that
for any $\lambda>0$, the following holds:
\[\mathbb P\left[\left|\sum_{i=1}^T \frac{X_i}{T} - \mu \right| \geq \lambda \right] \leq
2\exp\left(-2T^2\lambda^2/\sum_{i=1}^T(b_i-a_i)^2 \right).\]
For each sampled path $p_{s,l}$,
using any of the two before mentioned definitions of $\mathcal W$,
the following holds: $\mathbb P \left[p_{s,l} \right] \geq \mathcal W(p_{s,l})$.
This means each $\mathbold{c}_{\tau}$ is in the interval $[0,\alpha'(r)]$.
As a result, we can apply Hoeffding inequality on the random variables $\mathbold{c}_{\tau}$,
with $X_i=\mathbold{c}_{\tau}$, $T=\mathbold{\tau}$, $a_i=0$ and $b_i=\alpha'(r)$.
This yields
\begin{equation}
\label{eq:hoeffding2}
\mathbb P\left[ \left| \mathbold{c} -pc(r)\right| \geq \lambda \right] \leq
2\exp \left( \frac{-2{\mathbold{\tau}}\lambda^2}{ \alpha'(r)^2 } \right)  = \delta,
\end{equation}
which means if $\mathbold{\tau}\geq \frac{\alpha'(r)^2\log\frac{2}{\delta}}{2\lambda^2}$,
Algorithm~\ref{algorithm:kpath} estimates $k$-path score of $r$
within an additive error $\lambda$ with a probability at least $1-\delta$.
\end{proof}

The difference between the error bounds presented in Theorem~\ref{theorem:errorbound2}
and those that can be obtained for \textsf{RA-kpath} \cite{Alahakoon:2011:KCN:1989656.1989657} is that
on the one hand for the same error guarantee our algorithm 
requires $\alpha'(r)^2$ times less samples.
On the other hand, while in our algorithm each iteration (sample) takes $O(k|D|)$ time,
in \textsf{RA-kpath} it takes $O(k|V(G)|)$ time.
As a result, for the same error guarantee,
our algorithm is $O(\frac{|V(G)|}{\alpha'(r)^2|D|})$ times faster than \textsf{RA-kpath}.
Let $I$ be the set of the edges of $G$ whose both end-points are in $D$.
While space complexity of our algorithm is $O(|D|+|I|)$,
it is $O(|V(G)|+|E(G)|)$ for \textsf{RA-kpath}.
Note that for vertices of real-world graphs, $\alpha'$ is usually considerably less than $1$
and $|D|$ and $|I|$ are respectively considerably smaller than $|V(G)|$ and $|E(G)|$.
Similar results hold for our algorithm against the algorithm of \cite{Mahmoody:2016:SBC:2939672.2939869},
as it uses the same sampling strategy as \textsf{RA-kpath}.


In the last part of this section,
similar to the case of betweenness centrality,
in Theorem~\ref{theorem:errorbound3} we discuss how the stopping condition
of Algorithm~\ref{algorithm:kpath}
can be determined adaptively. 
The proof of this theorem is similar to the proof of Theorem~\ref{theorem:errorbound} (hence, we omit it).
There are, however, two key differences between 
Theorems~\ref{theorem:errorbound} and \ref{theorem:errorbound3}.
First, as shown in Lemma~\ref{lemma:variance2},
in Theorem~\ref{theorem:errorbound3}
the variance of each $\mathbold{c}_{\tau}$ is bounded by $\alpha'(r) pc(r)$,
where it is $\alpha(r)bc(r)$ in Theorem~\ref{theorem:errorbound}.
Second, as shown in Theorem~\ref{theorem:errorbound2},
when the number of samples is $\frac{\alpha'(r)^2\log\frac{4}{\delta}}{2\lambda^2}$
the estimation error of the $k$-path score of $r$ is bounded by $\lambda$, with
probability $1-\delta/2$. 
We refer to this quantity as $\omega'$ and use it in Theorem \ref{theorem:errorbound3}, as the 
upper bound, for the required number of samples.

\begin{lemma}
\label{lemma:variance2}
In Algorithm~\ref{algorithm:kpath},
for each $\mathbold{c}_{\tau}$ we have:
$\mathbb Var \left[\mathbold{c}_{\tau}\right] \leq \alpha'(r) pc(r).$
\end{lemma}
\begin{proof}
\begin{flalign}
&\mathbb Var \left[\mathbold{c}_{\tau}\right] = \mathbb E\left[\mathbold{c}_{\tau}^2 \right] - \mathbb E\left[\mathbold{c}_{\tau} \right]^2= \nonumber  \\  
& \sum_{s \in RF} \sum_{l \in [1,k]} \sum_{p_{s,l} \in D} \left( \frac{ |RF|^2 \mathcal W(p_{s,l})^2 \chi[r \in p_{s,l}]^2 }{|V(G)|^2\mathbb P\left[p_{s,l}\right]^2} \cdot  
\frac{\mathbb P\left[p_{s,l}\right]}{k|RF|} \right) - pc(r)^2  \nonumber \\
&\leq  \frac{\alpha'(r)}{k|V(G)|} \sum_{s \in RF} \sum_{l \in [1,k]} \sum_{p_{s,l} \in D}
\left( \mathcal W(p_{s,l}) \chi[r \in p_{s,l}] \right) - pc(r) ^2 \nonumber \\
&\leq \alpha'(r) pc(r) - pc(r) ^2 \leq \alpha'(r) pc(r). 
\end{flalign}
\end{proof}
Note that Lemma~\ref{lemma:variance2} holds for both definitions of $\mathcal W$.

In Algorithm~\ref{algorithm:kpath}, 
if the number of samples is equal to $\omega'$, 
with probability at least $1-\delta/2$ we have: $|pc(r) - \mathbold{c}| < \lambda $.
For the case of $\mathbold{\tau} < \omega'$,
we present the following theorem.

\begin{theorem}
\label{theorem:errorbound3}
In Algorithm~\ref{algorithm:kpath}, 
for real values $\delta_1,\delta_2 \in (0,1)$
and the value of $\omega'$ defined above,
after $\mathbold{\tau} < \omega'$ iterations of the loop in
Lines~\ref{line2:loop1:start}-\ref{line2:loop1:end}, we have:
\[\mathbb{P}\left[ pc(r)- \mathbold{c} \leq -A' 
\right] \leq \delta_1 \text{  and  }
\mathbb{P}\left[ pc(r)-\mathbold{c} \geq  B'
\right] \leq \delta_2,\]
where
\[A'=\frac{-1}{\mathbold{\tau}}  \log \frac{1}{\delta_1} 
\left( \frac{1}{3} - \frac{\omega' \alpha'(r)}{\mathbold{\tau}} + \sqrt{\left( \frac{1}{3} -\frac{\omega' \alpha'(r)}{\mathbold{\tau}} \right)^2  + \frac{2\mathbold{c}  \omega' \alpha'(r)}{\log\frac{1}{\delta_1}} }  \right)\]
\[B'=\frac{-1}{\mathbold{\tau}}  \log \frac{1}{\delta_1} 
\left( \frac{1}{3} - \frac{\omega' \alpha'(r)}{\mathbold{\tau}} + \sqrt{\left( \frac{1}{3} -\frac{\omega' \alpha'(r)}{\mathbold{\tau}} \right)^2  + \frac{2\mathbold{c}  \omega' \alpha'(r)}{\log\frac{1}{\delta_1}} }  \right).\]
\end{theorem}

\section{Experimental results}
\label{sec:experimentalresults}

We perform experiments
over several real-world networks to assess
the quantitative and qualitative behavior of
our proposed algorithms:
\textsf{ABAD} and \textsf{APAD}.
We test the algorithms over several real-world datasets
from the SNAP repository\footnote{\url{https://snap.stanford.edu/data/}},
including
the {\em com-amazon} network \cite{DBLP:conf/icdm/YangL12},
the {\em com-dblp} co-authorship network \cite{DBLP:conf/icdm/YangL12},
the {\em email-EuAll} email communication network \cite{DBLP:journals/tkdd/LeskovecKF07},
the {\em p2p-Gnutella31} peer-to-peer network \cite{DBLP:journals/tkdd/LeskovecKF07}
and
the {\em web-NotreDame} web graph \cite{albert1999dww}. 
All the networks are treated as directed graphs\footnote{The approximate algorithms can be successfully run
over larger networks.
Here, however, the bottleneck is to run the exact algorithm for measuring the errors!}.
For a vertex $r \in V(G)$, its {\em empirical approximation error} is defined as: 
$Error(r)=\frac{|app(r)-ext(r)| }{ext(r)} \times 100,$
where $app(r)$ is the calculated approximate score and
$ext(r)$ is its exact score.

%

\subsection{Betweenness centrality}
\label{sec:exp:bc}

In this section, we present the empirical results of \textsf{ABAD}.
We compare \textsf{ABAD} against the most efficient
existing algorithms, which are
\textsf{KADABRA} \cite{DBLP:conf/esa/BorassiN16}
and \textsf{BCD} \cite{DBLP:journals/corr/abs-1708-08739}.
The stopping condition used in \textsf{KADABRA} to determine the number of samples
can vary to either estimate betweenness centrality of (a subset of size $k$ of) vertices
or to find top $k$ vertices that have highest betweenness scores.
Here, due to consistency with the stopping condition used by Algorithm~\ref{algorithm:abad},
we use the setting of \textsf{KADABRA} that
estimates betweenness scores of (a subset of size $k$ of) vertices.
Furthermore, \textsf{KADABRA} has a parameter $k$ that determines
the number of vertices for which we want to estimate the betweenness score.
By increasing the value of $k$, \textsf{KADABRA} becomes slower.
In our experiments, we set\footnote{For a given value of $k$,
\textsf{KADABRA} estimates betweenness score of $k' \geq k$ vertices.
On the one hand, when we set $k$ to 1,
over several datasets our chosen vertices $r$ are not among the vertices
for which betweenness centrality is computed.
On the other hand, as mentioned before,
by increasing the value of $k$, \textsf{KADABRA} becomes slower.
Therefore, we set $k$ to 10 where 
in most cases our chosen vertices $r$ are among
the vertices for which \textsf{KADABRA} estimates betweenness centrality and yet,
it is not large.} $k$ to 10.
\textsf{BCD} is not adaptive
and it does not automatically compute the number of samples required for
the given values of $\lambda$ and $\delta$.
Furthermore, unlike \textsf{ABAD} and \textsf{KADABRA}, 
\textsf{BCD} is a (source) vertex sampler algorithm, hence, 
it requires much more time to process each sample. 
Thus, for \textsf{BCD} we cannot use the same number of samples as \textsf{ABAD} (or \textsf{KADABRA}).
In order to have fair comparisons, in our tests we let \textsf{BCD} have as much as samples that
it can process within the running time of \textsf{ABAD}. 
As a result, while \textsf{ABAD} and \textsf{BCD} will have (almost) the same running times,
the number of samples of \textsf{BCD} will be much less.

\begin{table*}
\caption{\label{table:com-amazon}Empirical evaluation of \textsf{ABAD} against \textsf{KADABRA} and \textsf{BCD}
for five vertices that have the highest betweenness scores.
The value of $\delta$ is $0.1$ and times are in seconds.
Times reported for \textsf{ABAD}
include both computing $\mathcal{RF}$ and $\mathcal{RT}$ and estimating betweenness score.
}
\begin{center}
\begin{tabular}{ l | l | l l l l l l | l l  l | l l }
\hline\hline 
\multirow{3}{*}{Dataset} & \multirow{3}{*}{$\lambda$} & \multicolumn{6}{|c|}{\textsf{ABAD}} & \multicolumn{3}{|c}{\textsf{KADABRA}} & \multicolumn{2}{|c}{\textsf{BCD}}\\ 
    &  &  \multicolumn{2}{c}{$\alpha$} &  \multicolumn{2}{c}{Time} &   \multicolumn{2}{c|}{Error} &  \multirow{2}{*}{Time} &  \multicolumn{2}{c|}{Error} &  \multicolumn{2}{c}{Error} \\
    &  & avg & max & avg & max  & avg & max  &    & avg & max  & avg & max   \\
\hline
\multirow{3}{*}{{\em com-amazon}} & 0.001   & 0.0000027  &  0.0000044 & 0.48 & 0.55 & \textbf{0.45} & \textbf{1.19} & 1.32 & 29.06 & 29.06  & 1.28 & 1.99 \\
 &0.00075 &    &  & 0.71 & 0.86 &  \textbf{0.36}  & \textbf{0.73 }  &  2.01 & 28.37 & 56.49 & 1.46 & 3.90 \\
 &0.0005  &   &   & 1.39 & 1.7  & \textbf{0.41} & \textbf{1.04} & 3.65 & 26.22 & 63.14 & 1.09 &   2.25 \\
\hline
\multirow{3}{*}{{\em com-dblp}} & 0.001   &   0.0077  &  0.0094  & 5.54 & 6.72 & \textbf{3.93} & \textbf{7.84} & 17.55 & 9.78 & 16.99  & 6.02 &  15.59  \\
& 0.00075 &   &  &  9.31 & 11.46 & \textbf{3.76}  & 7.08 &  28.79 &  9.16  & 14.31 & 4.46 & \textbf{6.42} \\
& 0.0005  &   &  & 19.07 & 23.15 & \textbf{3.41} & 7.46 & 54.46 &  8.49 & 14.98 & 3.47& \textbf{4.48} \\
\hline
\multirow{3}{*}{{\em email-EuAll}} & 0.001  &  0.0012 &  0.0013 & 1.50 & 1.79 & \textbf{1.10} &\textbf{2.01 } & 1.33 & 16.70 & 23.02& 2.63 & 6.60 \\
& 0.00075 &   &   &  2.47 & 2.91& \textbf{2.00}  & \textbf{3.73 }&  2.05 & 8.67  & 11.94 & 3.51 & 10.13 \\
& 0.0005  &    &   & 5.11 & 6.27 & 2.33   &  4.95 &  3.85  & 5.76 & 9.76  & \textbf{2.04 }& \textbf{3.89 } \\
\hline
\multirow{3}{*}{{\em p2p-Gnutella31}} & 0.001   &   0.01  & 0.01& 1.36  & 1.59&  \textbf{2.71 } & \textbf{4.24}& 12.44 &  4.94  & 13.98& 5.88  & 11.55\\
 & 0.00075 &      & & 2.27& 2.69 &  \textbf{2.36} & \textbf{4.32} & 19.71 &  2.93 & 7.46 & 6.18& 10.62 \\
 & 0.0005  &  &  &  4.75 & 5.57 &   \textbf{1.16} & \textbf{2.31}& 38.24 &  2.52 & 4.73  & 4.29& 9.28 \\
\hline
\multirow{3}{*}{{\em web-NotreDame}}  & 0.001   &   0.0021  &  0.0033& 0.91 & 1.25&  \textbf{2.17} & \textbf{4.45} & 2.80 & 3.59  & 7.14 & 4.59& 9.75 \\
 & 0.00075 &   &   &  1.5105 &  2.115 & \textbf{1.60}   & \textbf{2.81}& 4.47  & 4.21  & 7.85& 1.83 & 4.38  \\
 & 0.0005  &   &    &  2.98 &  4.11& \textbf{1.15} &  \textbf{2.38}  & 8.74 & 2.59& 4.83 & 2.27& 5.54 \\
\hline
\end{tabular}
\end{center}
\end{table*}

\textsf{ABAD} and \textsf{BCD} need to specify the vertex for which we want to estimate betweenness score.
For each dataset, we choose the five vertices that have the highest betweenness scores and
run \textsf{ABAD} and \textsf{BCD} for each of them.
Choosing the vertices that have the highest
betweenness scores has two reasons.
First, exact betweenness of the vertices that have a small $\alpha$ (or $|\mathcal{RF}|$) 
can be computed very efficiently \cite{DBLP:journals/corr/abs-1708-08739}.
Therefore, using approximate algorithms
for them is not very meaningful.
In fact by choosing vertices with high betweenness scores we guarantee 
that the chosen vertices have large enough $\alpha$ (or $|\mathcal{RF}|$),
so that it makes sense to run approximate algorithms for them.
Second, these vertices usually have a higher $\alpha$ than the other vertices.
Hence, it is likely that \textsf{ABAD} will estimate betweenness centrality of the other vertices
more efficiently.
For instance, in {\em web-NotreDame} consider those five randomly chosen vertices
that are used in the experiments of \cite{DBLP:journals/corr/abs-1708-08739}.
These vertices have the following $\alpha$ values:
$0.000000001$, $0.000011$, $0.0000016$, $0.00000051$ and $0.0000017$.
These values are much smaller than the $\alpha$ of
the vertices that have the highest betweenness scores (see Table~\ref{table:com-amazon}).
Therefore, \textsf{ABAD} will have a much better performance for these vertices.

Table~\ref{table:com-amazon} reports the experimental results.
In each experiment, the bold value shows the one that has the lowest error.
In all the experiments,
we set $\delta$ to $0.1$ for both
\textsf{ABAD} and \textsf{KADABRA}.
Since the behavior of these algorithms depends on the value of $\lambda$,
we run them for three different values of $\lambda$: $0.001$, $0.00075$ and $0.0005$.
Over {\em com-amazon}, some vertices 
are not among the vertices for which \textsf{KADABRA} computes the betweenness score,
hence, they do not contribute in the reported results.
Since the running time of \textsf{BCD} is very close to the running time of \textsf{ABAD} (as we set so),
we do not report it in Table~\ref{table:com-amazon}.

Our results show that \textsf{ABAD}
considerably outperforms both \textsf{KADABRA} and \textsf{BCD}.
It works always better than \textsf{KADABRA};
furthermore, in most cases (38 out of 42 cases)
it is more accurate than \textsf{BCD}.
This is due to very low values of $\alpha$ that vertices of real-world graphs have.
In our experiments, $\alpha$ is always smaller than $0.04$. 
This considerably restricts the search space and hence, improves the efficiency of \textsf{ABAD}.
The vertices tested in our experiments are those that have the highest
betweenness scores.
This implies that compared to the other vertices,
they usually have a high $\alpha$.
As a result, the relative efficiency of
\textsf{ABAD} with respect to \textsf{KADABRA} 
for the other vertices of the graphs (that do not have a very high betweenness score)
will be even better.
It should also be noted that for a vertex $r\in V(G)$,
while $\omega$ is computed,
the sets $\mathcal{RF}(r)$ and $\mathcal{RT}(r)$
are computed, too.
Therefore, computation of these sets does not impose any extra computational cost
on \textsf{ABAD}.

\begin{table*}
\caption{\label{table:exp:kpath}Empirical evaluation of \textsf{APAD} against \textsf{RA-kpath}.
All times are in seconds. Times reported for \textsf{APAD}
include both computing reachable set and estimating $k$-path score.
}
\begin{center}
\begin{tabular}{ l | l | l l l l l l | l l l }
\hline\hline 
\multirow{3}{*}{Dataset} & \multirow{3}{*}{$\#$samples} & \multicolumn{6}{|c|}{\textsf{APAD}} & \multicolumn{3}{|c}{\textsf{RA-kpath}} \\ 
     &  & \multicolumn{2}{c}{$\alpha'$} & \multicolumn{2}{c}{Time} &  \multicolumn{2}{c|}{Error} & \multirow{2}{*}{Time} & \multicolumn{2}{c}{Error} \\
     &  & avg & max & avg & max  & avg & max  &    &  avg & max   \\
\hline\hline
\multirow{3}{*}{{\em com-amazon}} & 50000  & 0.0176 & 0.0375 &  0.6202 & 0.6271   & \textbf{9.4273}  & \textbf{20.3553}   &    0.6201  & 17.3874  & 35.9696\\
 &100000  &      &      & 1.4710  & 1.4945 &  \textbf{8.7553} & \textbf{16.9069} &  1.4483   &  23.5959 & 47.8998 \\
 &500000  &    &      & 6.1657 & 6.3111 & \textbf{8.5740}  & \textbf{16.2172} &  6.0273   & 15.0688  & 28.1798 \\
\hline
\multirow{3}{*}{{\em com-dblp}} & 50000   & 0.2984  & 0.5837  &  0.6120 & 1.4015 & \textbf{8.3598} & \textbf{14.0485}   & 0.6990  &   18.8275& 39.3276  \\
& 100000 &    &     &  1.4015 &  1.4227 &  \textbf{8.3598}   &  \textbf{13.5833} & 1.3864   &  18.8275 & 56.3778  \\
& 500000  &   &      & 6.0236 &  6.1431 & 7.3327  & \textbf{10.8506} & 5.8907   &  \textbf{7.1508} &  12.7557  \\
\hline
\multirow{3}{*}{{\em email-EuAll}} & 50000   &  0.3057  & 0.4563 & 0.4517  & 0.4570 & \textbf{4.4155}  &  \textbf{9.9514}  & 0.4409   & 9.0128  & 26.2509 \\
& 100000 &     &      &  1.0388 & 1.0904  &  \textbf{3.2181} &  \textbf{7.6999} &  1.0128  &  5.4020 &  13.1260 \\
& 500000  &   &     &  4.3868 &  4.4107 & \textbf{1.3479} &  \textbf{3.8887} &  4.3044    &  2.3709 & 3.8737 \\
\hline
\multirow{3}{*}{{\em p2p-Gnutella31}} & 50000   &  0.2050  & 0.5776  & 0.1666 & 0.1715 & \textbf{10.1815} & \textbf{26.4400} & 0.1633 & 25.1530& 56.1906    \\
 & 100000 &    &     &  0.7527& 0.3947  &  \textbf{9.8433} & \textbf{24.2559}  &  0.3436 &  12.8763 &   24.9524 \\
 & 500000  &     &      & 1.5625  & 1.5896 & \textbf{6.5869} & \textbf{20.6747} & 1.5652  &  20.1032 &  32.5738 \\
\hline
\multirow{3}{*}{{\em web-NotreDame}}  & 50000   &   0.3520 &  0.7182  &  0.6116  &  0.6152  & \textbf{4.6194} & \textbf{9.4785} &  0.6033  & 7.8810 & 13.2039  \\
 & 100000 &   &    &  1.2208  & 1.2331   & 3.4717  & 11.1948  &  1.2108  & \textbf{2.3478}  &  \textbf{5.8867} \\
 & 500000  &    &     &  5.9969  & 6.0496  & 2.9064 & \textbf{4.9510} & 6.0111 & \textbf{2.2955} & 4.9636  \\
\hline
\end{tabular}
\end{center}
\end{table*}

Someone may complain that \textsf{ABAD} (like \textsf{BCD}) estimates 
betweenness score of only one vertex,
whereas \textsf{KADABRA} can estimate betweenness scores of all vertices.
However, as mentioned before,
in several applications it is sufficient to compute this index for only one vertex 
or for a few vertices.
Even if in an application we require to compute betweenness score of a set of vertices,
\textsf{ABAD} can be still useful.
In this case, after computing $\mathcal{RF}$ and $\mathcal{RT}$ of all the vertices in the set in parallel,
someone may run the sampling part of Algorithm~\ref{algorithm:abad}  
(Lines~\ref{line:loop1:start}-\ref{line:loop1:end} of the algorithm)
for all the vertices in the set in parallel,
as was done in the existing algorithms.
Note that even if \textsf{ABAD} is run independently for each vertex in the set,
in many cases it will outperform \textsf{KADABRA}.
For example, consider our experiments where we run \textsf{ABAD}
for five vertices that have the highest betweenness scores.
In the reported experiments,
for each dataset,
by simply summing up the running times and the number of samples of five vertices,
we can obtain the running time and the number of samples
of the independent runs of \textsf{ABAD} for all the five vertices.
Even in this case,
the number of samples used by \textsf{ABAD} to estimate betweenness centrality
of all five vertices is considerably (at least 10 times!) less than the number of samples used
by \textsf{KADABRA}.
Moreover, in most cases, \textsf{ABAD} is more accurate than \textsf{KADABRA}.
Finally, while over {\em p2p-Gnutella31}
\textsf{ABAD} takes less time to compute betweenness scores of all five vertices,
over the other datasets \textsf{KADABRA} is faster.
A reason for the faster performance of \textsf{KADABRA}
over
{\em com-amazon},
{\em com-dblp}, {\em email-EuAll}
and {\em web-NotreDame} (despite the fact that it uses much more samples) is that it
processes most of these samples in parallel,
whereas in the independent runs of \textsf{ABAD},
less samples are processed in parallel.
Using considerably less samples by \textsf{ABAD} gives a sign that if
\textsf{ABAD} estimates betweenness scores of all vertices of a set in parallel,
it might become always faster than \textsf{KADABRA}.
Developing a fully parallel version of \textsf{ABAD} that estimates betweenness scores of
a set of vertices in parallel is an interesting direction for future work.

\subsection{$k$-path centrality}

In this section, we present the empirical results of \textsf{APAD}.
We compare \textsf{APAD} against \textsf{RA-kpath} \cite{Alahakoon:2011:KCN:1989656.1989657}
that samples a vertex $s \in V(G)$, an integer $l \in \left[1 , k\right]$
and a random path that
starts from $s$ and has $l$ edges.
This sampling method has been used by
some other work, including \cite{Mahmoody:2016:SBC:2939672.2939869}.
The error bound presented in \cite{Alahakoon:2011:KCN:1989656.1989657}
is not adaptive and has only one parameter $\alpha$ (different from the parameter $\alpha$ used in this paper),
instead of two parameters $\lambda$ and $\delta$ that we have for \textsf{APAD}.
This makes comparison of the methods slightly challenging.
In order to have a fair comparison
(regardless of the method used to determine the number of samples
for an ($\lambda,\delta$)-approximation),
we run each of them for a fixed number of iterations and report 
the empirical approximation error and running time.
\textsf{APAD} requires to specify the vertex $r$ for which we want to compute $k$-path centrality.
In order to have experiments consistent with Section~\ref{sec:exp:bc},
in this section we use the same vertices used in the evaluation of \textsf{ABAD}. 
We repeat each experiment for three times and report the average results.

Table~\ref{table:exp:kpath} reports the experimental results of $k$-path centrality.
In each experiment, the bold value shows the one that has the lowest (average/maximum) error.
As can be seen in the table, in most cases \textsf{APAD} shows a better
average error and maximum error than \textsf{RA-kpath}.
This is particularly more evident over the datasets (e.g., \textit{com-amazon}) that have a low $\alpha'$ value,
which empirically validates our theoretical analysis on the connection between 
$\alpha'$ and the efficiency of \textsf{APAD}.
Note that as before mentioned, \textsf{RA-kpath} computes $k$-path centrality of all vertices,
whereas \textsf{APAD} requires to specify the vertex for which we want to compute the centrality score.
As a results, for different given vertices,
while only one run of \textsf{RA-kpath} is sufficient,
we require to have different runs of \textsf{APAD}.
Therefore in Table~\ref{table:exp:kpath}, for each dataset and sample size
we report only one time for \textsf{RA-kpath}.
In contrast, for each dataset and sample size we have five running times for \textsf{APAD},
where the average and maximum values are reported in Table~\ref{table:exp:kpath}.
The times reported for \textsf{APAD} include both
computing reachable set and estimating $k$-path score.
These times are only slightly larger than the times reported for \textsf{RA-kpath}.
This means that computing the reachable set of a vertex can be done 
very efficiently and in a time negligible compared to the running time of the whole process.

\section{Conclusion}
\label{sec:conclusion}

In this paper,
first we presented a novel adaptive algorithm
for estimating betweenness score of a vertex $r$ in a directed graph $G$.
Our algorithm computes the sets $\mathcal{RF}(r)$ and $\mathcal{RT}(r)$,
samples from them,
and stops as soon as some condition is satisfied.
The stopping condition depends on the samples met so far, $|\mathcal{RF}(r)|$ and $|\mathcal{RT}(r)|$.
We showed that our algorithm gives a more efficient $(\lambda,\delta)$-approximation
than the existing algorithms.
Then, we proposed a novel algorithm for estimating $k$-path centrality of $r$ 
and showed that in order to give a $(\lambda,\delta)$-approximation,
it requires considerably less
samples. Moreover, it processes each sample faster and with less memory.
Finally, we empirically evaluated our centrality estimation algorithms and showed
their superior performance.



%

\bibliographystyle{plain}
\bibliography{allpapers}

\end{document}